\documentclass{llncs}
\usepackage{llncsdoc}
\usepackage{graphicx}
\usepackage{stmaryrd,color}

\usepackage{amsmath}
\begin{document}
\title{Pumping lemma and Ogden lemma for displacement context-free grammars}
\author{Alexey Sorokin \inst{1,2}}
\institute{Moscow State University, Faculty of Mathematics and Mechanics
\and Moscow Institute of Physics and Technology, \\Faculty of Innovations and High Technologies}
\maketitle
\begin{abstract}
The pumping lemma and Ogden lemma offer a powerful method to prove that a particular language is not context-free. In 2008 Kanazawa proved an analogue of pumping lemma for well-nested multiple-context free languages. However, the statement of lemma is too weak for practical usage. We prove a stronger variant of pumping lemma and an analogue of Ogden lemma for this language family. We also use these statements to prove that some natural context-sensitive languages cannot be generated by tree-adjoining grammars.
\end{abstract}

\newcommand{\inbr}[1]{\{ #1 \}}
\newcommand{\kav}[1]{``#1''}
\newcommand{\inang}[1]{\langle #1 \rangle}
\newcommand{\hide}[1]{}

\newcommand{\aar}{\mathrm{ar}}
\newcommand{\Bb}{\mathcal{B}}
\newcommand{\Aa}{\mathcal{A}}

\section{Introduction}

Since $80$-s context-free grammars are known to be too restrictive for syntactic description of natural language (\cite{Shieber1987}). The class of mildly context-sensitive languages (\cite{Joshi1985}) was an informal attempt to capture the degree of context-sensitivity required for most common language phenomena keeping as much advantages of context-free grammars as possible. The principal properties to inherit are the feasible polynomial parsing complexity, independence of derivation from context (the notion of context had to be extended to handle long-distance dependencies) and existence of convenient normal forms. The class of well-nested multiple context-free languages (wMCFLs) is one of the candidates to satisfy these requirements\footnote{see \cite{KanazawaSalvati2012} for discussion.}. The corresponding grammar formalism, well-nested multiple context-free grammars (well-nested MCFGs or wMCFGs), is defined as a subclass of multiple context-free grammars (MCFGs, \cite{Seki1991}) with rules of special form providing the correct embedding of constituents. In particular, $2$-wMCFGs are equivalent to tree-adjoining grammars (TAGs, \cite{VijaySchanker1986}, \cite{JoshiSchabes1997}) and then to head grammars (\cite{Pollard1984}).

We find it sensible to think of wMCFGs as the generalization of head grammars, not the restriction of MCFGs. Our approach is based on two principal ideas. The first is to derive not words but terms whose values are the words of the language. Thus the generative power of a grammar formalism essentially depends on the set of term connectives and their interpretation as language operations. If the only operation we use is concatenation, the terms are just strings of terminals and nonterminals and we get nothing but context-free grammars. Our approach seems redundant there, but is vital in more complex cases. The second idea is to extend the alphabet by a distinguished separator $1$\footnote{This idea is inspired by the works of Morrill and Valent{\'\i}n on discontinuous Lambek calculus (\cite{Morrill2011}, \cite{Valentin2012}).}. Using the separators, the rules of well-nested MCFGs may be simulated with the help of intercalation connectives. The binary operation $\odot_j$ of $j$-intercalation replaces the $j$-th separator in its first argument by its second argument (for example, $a1b1c \odot_2 a1b = a1ba1bc$). It is straightforward to prove that all \kav{well-nested} combinations of constituents can be presented using only intercalation and concatenation operations.	

The exact generative power of wMCFGs is not known. Moreover, some languages are supposed not to be wMCFLs, although they are not proved to be outside this family. The most known example is the MIX language $\inbr{w \in \inbr{a, b, c}^* \mid |w|_a = |w|_b = |w|_c}$. It was shown in \cite{KanazawaSalvati2012} to be not a $2$-MCFL, but the proof used combinatorial and geometric arguments arguments which are troublesome to be generalized for the class of all wMCFGs. The pumping lemma for wMCFLs presented in \cite{Kanazawa2009} is also too weak since it does not impose any conditions on the length and position of the pumped segment. We prove a stronger version of pumping lemma and weak Ogden lemma\footnote{A stronger version of Ogden lemma for tree-adjoining languages which form the first level of well-nested MCFL hierarchy was proved in \cite{PalisShende1995}, but the proof is difficult to be generalized.} for well-nested MCFGs basing on the ideas already used in \cite{Kanazawa2009}. Our variant of Ogden lemma allows us to give a simple proof of the fact that MIX cannot be generated by a TAG.

We suppose the reader to be familiar with the basics of formal languages theory nevertheless all the required definitions are explicitly formulated. In order not to interrupt the flow of the paper heavy technical proofs are left in the Appendices in the end of the paper.

\section{Preliminaries}
\subsection{Terms and their equivalence}
In this section we define displacement context-free grammars (DCFGs) which are a more \kav{purely logical} reformulation of well-nested MCFGs. The first subsection is devoted to the notions of term, context and generalized context that play the key role in the architecture of DCFGs, it also contains some results on term equivalence which are extensively used in the further. We mostly follow the definitions from \cite{Sorokin2013LFCS}, but the purposes of this work require some technical complications.

Let $\Sigma$ be a finite alphabet and $1 \notin \Sigma$ be a distinguished separator, let $\Sigma_1 = \Sigma \cup \{1\}$. For every word $w \in \Sigma_1^*$ we define its rank $rk(w) = |w|_1$. We define the $j$-th intercalation operation $\odot_j$ which consists in replacing the $j$-th separator in its first argument by its second argument. For example, $a1b11d\odot_2 c1c = a1bc1c1d$.

Let $k$ be a natural number and $N$ be the set of nonterminals. The function $rk \colon N \to \overline{0,k}$ assigns every element of $N$ its rank. Let $Op_k = \{ \cdot, \odot_1, \ldots, \odot_k\}$ be the set of binary operation symbols, then the ranked set of $k$-correct terms $Tm_k(N, \Sigma)$ is defined in the following way (we write simply $Tm_k$ in the further):

\begin{enumerate}
    \item $N \subset Tm_k(N, \Sigma)$,
    \item $\Sigma^* \subset Tm_k(N, \Sigma), \: \forall w \in \Sigma^* \: rk(w) = 0$,
    \item $1 \in Tm_k, \: rk(1) = 1$,
    \item If $\alpha, \beta \in Tm_k$ and $rk(\alpha) + rk(\beta) \leq k$, then $(\alpha \cdot \beta) \in Tm_k, \\ rk(\alpha \cdot \beta) = rk(\alpha) + rk(\beta)$.
    \item If $j \leq k, \: \alpha, \beta \in Tm_k, \: rk(\alpha) + rk(\beta) \leq k + 1, \: rk(\alpha) \geq j$, then \\ $(\alpha \odot_j \beta) \in Tm_k, \: rk(\alpha \cdot \beta) = rk(\alpha) + rk(\beta) - 1$.
\end{enumerate}

We refer to the elements of the set $N \cup \Sigma^* \cup \{1\}$ as basic subterms. We will often omit the symbol of concatenation and assume that concatenation has greater priority then intercalation, so $Ab \odot_2 cD$ means $(A \cdot b) \odot_2 (c \cdot D)$. This simplification allows us to consider words in the alphabet $\Sigma^*_1$ as terms either. The set of $k$-correct terms includes all the terms of sort $k$ or less that also do not contain subterms of rank greater than $k$.

\newcommand{\Vvar}{\mathrm{Var}}
\newcommand{\rk}{\mathrm{rk}}
Let $\Vvar = \inbr{x_1, x_2, \ldots}$ be a countable of variables. We assume that every variable has a fixed rank and there is an infinite set of variables of every rank. A context $C[x]$ is a term where a variable $x$ occurs in a leaf position, the rank of $x$ must respect the constraints of term construction. Provided $\beta \in Tm_k$ and $\rk(x) = \rk(\beta)$, $C[\beta]$ denotes the result of substituting $\beta$ for $x$ in $C$. For example, $C[x] = b1 \odot_1 (a \cdot x)$ is a context and $C[A \cdot c] = b1 \odot_1 aAc$. The notion of multicontext is defined in the same way, except it may contain several distinct variables $x_1, \ldots, x_t$. In the case $t=0$ a multicontext is just a term. If for any $i$ it holds that $\rk(\alpha_i) = \rk(x_i)$, $C[\alpha_1, \ldots, \alpha_t]$ denotes the result of substituting $\alpha_1, \ldots, \alpha_t$ for $x_1, \ldots, x_t$ in $C$.

We call a term (respectively, a context, a multicontext) ground if it contains no occurrences of nonterminals. let $\mu$ be a valuation function, mapping every variable of rank $l$ to some language of words of rank $l$. Then every ground multicontext $\alpha$ is assigned a value, interpreting the elements of $\Sigma_1^*$ as themselves and the connectives from $Op_k$ as corresponding language operations. Note that ground terms have the same value under all valuations. Two ground multicontexts $C_1[x_1, \ldots, x_t]$ and $C_2[x_1, \ldots, x_t]$ with the same variables are equivalent, if the expressions $C_1[\mu(x_1), \ldots, \mu(x_t)]$ and $C_2[\mu(x_1), \ldots, \mu(x_t)]$ have the same value under any valuation $\mu$. The equivalence relation is denoted by $\sim$, note that $\alpha \sim \mu(\alpha)$ for any ground term $\alpha$. $\sim$ is a congruence relation, which means that the equivalences $C' \sim C''$ and $\alpha_i \sim \beta_i$ for any $i \leq t$ imply $C'[\alpha_1, \ldots, \alpha_t] \sim C''[\beta_1, \ldots, \beta_t]$. The lemma below follows from the definitions\footnote{As shown in \cite{Valentin2012}, these equivalencies form the axiomatics of equational theory of displacement algebras}.
\begin{lemma}\label{term-eq-basic}
The following ground multicontexts are equivalent:
\begin{enumerate}
\item $(x_1 \cdot x_2) \cdot x_3 \sim x_1 \cdot (x_2 \cdot x_3)$,
\item $(x_1 \cdot x_2) \odot_j x_3 \sim (x_1 \odot_j x_3) \cdot x_2$ if $j \leq \rk(x_1)$,
\item $(x_1 \cdot x_2) \odot_j x_3 \sim x_1 \cdot (x_2 \odot_{j - \rk(x_1)} x_3)$ if $\rk(x_1) < j \leq \rk(x_1) + \rk(x_2)$,
\item $(x_1 \odot_l x_2) \odot_j x_3 \sim (x_1 \odot_j x_3) \odot_{l + \rk(x_3) - 1} x_2$ if $j < l$,
\item $(x_1 \odot_l x_2) \odot_j x_3 \sim x_1 \odot_l (x_2 \odot_{j - l + 1} x_3)$ if $l \leq j < l + \rk(x_2)$,
\item $(x_1 \odot_l x_2) \odot_j x_3 \sim (x_1 \odot_{j - \rk(x_2) + 1} x_3) \odot_l x_2$ if $j \geq l + \rk(x_2)$.
\item $1 \odot_1 x_1 \sim x_1$,
\item $x_1 \odot_j 1 \sim x_1$ for any $j \leq \rk(x_1)$.
\end{enumerate}
\end{lemma}

Let $\alpha$ be a term, we call its skeleton a ground multicontext $C_{\alpha}[x_1, \ldots, x_t]$ such that $\alpha = C_{\alpha}[B_1, \ldots,B_t]$ for some nonterminals $B_1, \ldots, B_t$. A skeleton is obtained by replacing all the nonterminal leaves of $\alpha$ by variables of the same rank in left-to-right order. Two terms $\alpha_1$ and $\alpha_2$ are called equivalent if they can be represented in the form $\alpha_1 = C_1[A_1, \ldots, A_t]$ and $\alpha_2 = C_2[A_1, \ldots, A_t]$ for some equivalent ground multicontexts $C_1$ and $C_2$.

With every multicontext $\alpha$ we associate its syntactic tree $tree(\alpha)$ in a natural way. Then submulticontexts of $\alpha$ correspond to the nodes of this tree and vice versa, a submulticontext is internal if it corresponds to an internal node (it means the submulticontext contains a binary connective). A multicontext is $k$-essential if its rank is less than $k$, as well as the rank of all the variables and nonterminals occurring in it. The next lemma is proved in the Appendix \ref{app-term-eq}.

\begin{lemma}\label{multicontext-eq}
For any $k$-essential multicontext $C$ there is an equivalent $k$-correct multicontext $C'$.
\end{lemma}

Since a term is just a special case of a multicontext, the next corollary holds:

\begin{corollary}\label{term-eq}
For any $k$-essential term $\alpha$ exists an equivalent $k$-correct term $\alpha'$.
\end{corollary}

\subsection{Displacement context-free grammars}

This subsection introduces the notion of a displacement context-free grammar. In the definitions below $GrTm_k$ denotes the set of all ground terms in $Tm_k$.

\begin{definition}
A $k$-displacement context-free grammar ($k$-DCFG) is a quadruple $G = \inang{ N, \Sigma, P, S }$, where $\Sigma$ is a finite alphabet, $N$ is a finite ranked set of nonterminals and $\Sigma \cap N = \emptyset, S \in N$ is a start symbol such that $rk(S) = 0$ and $P$ is a set of rules of the form $A \to \alpha$. Here $A$ is a nonterminal, $\alpha$ is a term from $Tm_k(N, \Sigma)$, such that $rk(A) = rk(\alpha)$.
\end{definition}

\begin{definition}
The derivability relation $\vdash_G \in N \times Tm_k$ associated with the grammar $G$ is the smallest reflexive transitive relation such that the facts $(B \to \beta) \in P$ and $A \vdash C[B]$ imply that $A \vdash C[\beta]$ for any context $C$. Let the set of words derivable from $A \in N$ be $L_G(A) = \{ \nu(\alpha) \mid A \vdash_G \alpha, \: \alpha \in GrTm_k\}$, then $L(G) = L_G(S)$.
\end{definition}

\begin{example}\label{dcfg-example}
Let the $i$-DCFG $G_i$ be the grammar $G_i = \inang{ \{S, T\}, \{a, b\}, P_i, S}$. Here $P_i$ is the following set of rules (notation $A \to \alpha | \beta$ means $A \to \alpha, A \to \beta)$:
$$\begin{array}{rcl}
    S & \to & \underbrace{(\ldots(}_{i - 1\mbox{ times}} aT \odot_1 a) + \ldots) \odot_1 a \; | \; \underbrace{(\ldots(}_{i - 1\mbox{ times}} bT \odot_1 b) + \ldots) \odot_1 b\\
    T & \to & \underbrace{(\ldots(}_{i - 1\mbox{ times}} aT \odot_1 1a) + \ldots) \odot_{i} 1a \; | \; \underbrace{(\ldots(}_{i - 1\mbox{ times}} bT \odot_1 1b) + \ldots) \odot_{i} 1b \; | \; 1^i
\end{array}$$
$G_i$ generates the language $\{ w^{i + 1} \mid w \in \{a, b\}^+\}$. For example, this is the derivation of the word $(aba)^3$ in $G_2$: $S \to (aT \odot_1 a) \odot_1 a \to (a((bT \odot_1 1b) \odot_2 1b) \odot_1 a) \odot_1 a \to (a((b((aT \odot_1 1a) \odot_2 1a) \odot_1 1b) \odot_2 1b)\odot_1 a) \odot_1 a \to (a((b((a11 \odot_1 1a) \odot_2 1a) \odot_1 1b) \odot_2 1b)\odot_1 a) \odot_1 a = (a(b(a1a1a \odot_1 1b) +_2 1b)+_1 a) \odot_1 a = (aba1ba1ba \odot_1 a) \odot_1 a = abaabaaba$.
\end{example}

Two $k$-DCFGs are equivalent if they generate the same language. Since internal nodes of terms in a $k$-DCFG rules are also of rank $k$ or less, the $k$-DCFGs can be binarized just like the context-free grammars to obtain a variant of Chomsky normal form. Precisely, the following theorem holds (see \cite{Sorokin2013LFCS} for details):

\begin{theorem}
Every $k$-DCFG is equivalent to some $k$-DCFG $G = \inang{N, \Sigma, P, S}$ which has the rules only of the following form:
\begin{enumerate}
\item $A \to B \cdot C,  \mbox{ where } A \in N,\: B, C \in N - \{S\}$,
\item $A \to B \odot_j C,\mbox{ where } j \leq k,\: A \in N,\: B, C \in N - \{S\}$,
\item $A \to a,\mbox{ where }  A \in N,\: a \in \Sigma_1$,
\item $S \to \epsilon$.
\end{enumerate}
\end{theorem}

We have already mentioned that $k$-DCFGs are equivalent to $(k+1)$-wMCFGs. In the case of $k=1$ this statement is straightforward since both $1$-DCFGs and $2$-wMCFGs are just reformulations of Pollard wrap grammars (\cite{Pollard1984}). We will not recall the definitions of a wMCFG, the interested reader may consult \cite{Seki1991} and \cite{Kanazawa2009}.

\section{Terms and derivations in DCFGs}

In this section we investigate more thoroughly the properties of terms and derivation in DCFGs. At first we give some fundamental notions. We assume that all the grammars are in Chomsky normal form.

\begin{definition}
A node $v'$ in the syntactic tree is a direct descendant of a node $v$ if $rk(v') = rk(v)$, $v'$ is a descendant of $v$ and all the nodes on the path from $v$ to $v'$ has the same rank as $v$ and $v'$. A subterm $\beta$ is a direct subterm of a term $\alpha$, if its root node is the direct descendant of the root of $\alpha$.
\end{definition}

Let $\alpha$ be a term of rank $l$, we denote by $\alpha \otimes\footnote{This notation is brought from discontinuous Lambek calculus.} (u_1, \ldots, u_l)$ the result of simultaneous replacement of all the separators in $\alpha$ by $u_1, \ldots, u_l$.

\begin{lemma}
Let $\alpha = C[\beta]$ for some ground context $C$ and term $\beta$ of rank $l$. There exist words $s_1, s_2, u_1, \ldots, u_l \in \Sigma_1^*$ depending only from the context $C$ such that $\alpha \sim s_1 (\beta \otimes (u_1, \ldots, u_l)) s_2$ and $\rk(\alpha) = \rk(s_1) + \rk(s_2) + \sum\limits_{i=1}^{l} \rk(u_i)$.
\end{lemma}
\begin{proof}
Induction on the structure of the context $C$. The induction step uses Lemma \ref{term-eq-basic} and the equivalence between a ground term and its value.
\end{proof}

\begin{lemma}\label{direct-synt}
Let $\beta$ be a direct subterm of a term $\alpha$ and $C$ be the ground context such that $\alpha = C[\beta]$.
Then the equivalence $\alpha \sim s_1 (\beta \otimes (y_1 1 z_1, \ldots, y_l 1 z_l)) s_2$ holds for some words $s_1, s_2, y_1, z_1 \ldots, y_l, z_l \in \Sigma^*$, depending only from the context $C$.
\end{lemma}
\begin{proof}
Induction on the structure of the context $C$, the base is trivial. On the induction step consider the root connective of the term $\alpha$. If this connective is $\cdot$, then $\alpha$ has the form $\alpha' \cdot \eta$ or $\eta \cdot \alpha'$ for some ground term $\eta$ of sort $0$ and some term $\alpha$ such that $\beta$ is its direct subterm. The statement follows from the induction hypothesis with the help of the fact that $\eta$ is equivalent to the word $\nu(\eta) \in \Sigma^*$. 

If the root connective is $\odot$, then $\alpha = \alpha' \odot_j \eta$ or $\alpha = \eta \odot_1 \alpha'$ for some ground term $\eta$ of sort $1$ and $\alpha'$ having a direct subterm $\beta$. Then the statement also easily follows from the induction hypothesis.
\end{proof}

Let $D$ be the derivation of $\alpha$ from some nonterminal $A$ of the grammar $G$ (we denote it by $D \colon A \vdash \alpha$). We associate with $D$ its derivation tree $T_D$ obtained by attaching nonterminals to the nodes of $tree(\alpha)$. The labeling procedure is the following: if the last step of $D$ applied the rule $B \to \beta$ in the context $C$ then we label by $B$ the root node of the inserted subtree and keep other labels unchanged. Since $G$ is in Chomsky normal form, only the nonterminal leaves of $tree(\alpha)$ are unlabeled. Then we label every such node by the nonterminal it contains. The lemma below is proved by induction on derivation length.

\newcommand{\deriv}[3]{#1 \colon #2 \vdash #3}
\begin{lemma}\label{deriv-decompose}
\leavevmode
\begin{enumerate}
\item Let $D \colon A \vdash \alpha$ and $T_D$ be the corresponding derivation tree. For every representation $\alpha = C[\beta]$ there are derivations $\deriv{D_1}{A}{C[B]}$ and $\deriv{D_2}{B}{\beta}$ such that $T_D$ is obtained by replacing $B$ with $T_{D_2}$ in the context $C$.
\item Let $D \colon A \vdash \alpha$ and $T_D$ be the corresponding derivation tree. For every representation $\alpha = C[\beta_1, \ldots, \beta_t]$ there are derivations $\deriv{D_0}{A}{C[B_1, \ldots B_t]}$ and $\deriv{D_i}{B_i}{\beta_i}$ for any $i \leq t$ such that $T_D$ is obtained by replacing $B_i$ with $T_{D_i}$ in the multicontext $C$.
\end{enumerate}
\end{lemma}

A rule $A \to \alpha$ is derivable in a grammar $G$ if $A \vdash_G \alpha$. Adding derivable rules to a grammar does not change the language it generates. Rules $A \to \alpha$ and $A \to \alpha'$ are called equivalent if the terms $\alpha$ and $\alpha'$ are equivalent. If one of such rules is already in $G$, adding the other does not affect the generated language. Note that if every rule of $G'$ is equivalent to some rule of $G$ and vice versa, then the grammars themselves are also equivalent.

We call a term $\alpha$ derivable in the grammar $G$ if $A \vdash \alpha$ for some nonterminal $A$ and $S$-derivable, if it is derived from initial nonterminal. Let $T$ be a derivation tree for the derivation $\deriv{D}{A}{\alpha}$, we call its subtree $T'$ inherent if every node in $T'$ either have the same number of children as in $T$ or has no children at all.

Let us consider inherent subtrees more attentively. As any derivation tree, an internal subtree $T'$ with a root labeled by $B$ may be a considered as a syntactic tree of some term $\beta$, in this case it holds that $B \vdash \beta$. If $T'$ contain $t$ nonterminal nodes, then there is a representation $\beta = C[B_1, \ldots, B_t]$ such that $\alpha = C_0[C[\beta_1, \ldots, \beta_t]]$ for some context $C_0$, multicontext $C$, nonterminals $B_1, \ldots, B_t$ and terms $\beta_1, \ldots, \beta_t$, satisfying the following properties:
\begin{enumerate}
\item $A \vdash C_0[B]$,
\item $B \vdash C[B_1, \ldots B_t]$,
\item $B_i \vdash \beta_i$ for any $i \leq t$.
\end{enumerate}

Let us consider inherent subtrees more attentively. As any derivation tree, an internal subtree $T'$ with a root labeled by $B$ may be a considered as a syntactic tree of some term $\beta$, in this case it holds that $B \vdash \beta$. Let $T'$ contain $t$ nonterminal nodes and $C$ be its skeleton, then there is a representation $\beta = C[B_1, \ldots, B_t]$ such that $\alpha = C_0[C[\beta_1, \ldots, \beta_t]]$ for some context $C_0$, nonterminals $B_1, \ldots, B_t$ and terms $\beta_1, \ldots, \beta_t$, satisfying the following properties:
\begin{enumerate}
\item $A \vdash C_0[B]$,
\item $B \vdash C[B_1, \ldots B_t]$,
\item $B_i \vdash \beta_i$ for any $i \leq t$.
\end{enumerate}

Let $v$ be a node of rank $l$ in the derivation tree $T$. We call the vicinity of $v$ an inherent subtree $U_v$ satisfying the following properties: $v$ is a node of $U_v$, all the leaves of $U_v$ are of rank distinct from $l$ or are the leaves of the whole tree $T$, the root of $U_v$ is not of rank $l$ or is the root of the whole tree, all internal nodes of $U_v$ are of rank $l$. By definition every node has a unique vicinity. Note that the vicinities of two nodes cannot have common internal nodes unless these vicinities coincide.

Let $G$ be a $k$-DCFG containing $N_l$ nonterminals of rank $l$ and $T$ be a derivation tree in this grammar. We call an $l$-matryoshka\footnote{Matryoshka is a Russian souvenir consisting of several dolls nested one into another. We use this term since the yields of the subtrees, whose roots are the elements of an $l$-matryoshka, demonstrate the same nesting property.} a subbranch of length $N_l+1$ or more, containing only nodes of rank $l$. Note that all the elements of $l$-matryoshka are direct descendants of each other. By the pigeon-hole principle it contains two nodes with the same nonterminal label.

We denote the depth of a term $\beta$ by $d(\beta)$. A term is called $l$-internal if all its internal nodes, possibly except the root, are of rank $l$. If it is additionally $l$-essential and $d(\beta) \leq N_l+1$, then it is called $l$-redundant. The grammar $G$ is called $l$-duplicated, if for every derivable rule $A \to \alpha$ with $\alpha$ being $l$-redundant, there is an equivalent derivable rule $A \to \alpha'$ with $(l-1)$-correct term $\alpha'$.

\begin{lemma}\label{l-duplication}
For every $k$-DCFG $G$ in Chomsky normal form and every $l \leq k$ there is an equivalent $l$-duplicated grammar $G'$ in Chomsky normal form with the same set of nonterminals of rank $l$ and greater.
\end{lemma}
\begin{proof}
We call a rule $A \to \alpha$ with $l$-redundant term $\alpha$ unduplicated, if there is no equivalent rule $A \alpha'$ for an $(l-1)$-correct term $\alpha'$ in $G$. Since $\alpha$ is $l$-redundant, by Lemma \ref{term-eq} there is an equivalent $(l-1)$ correct term $\alpha'$. We enrich the set of rules with productions, obtained during binarization of the rule $A \to \alpha'$, thus the rule $A \to \alpha$ is not unduplicated anymore. Let us prove no new unduplicated rule has appeared. Indeed, if a term is in the right side of such derivable rule, then all its binary nodes except the root are of rank $l$. It means that there is at least one nonterminal of rank $l$ in every rule used in its derivation. But all new rules do not contain nonterminals of rank $l$ and greater since $\alpha'$ is an $(l-1)$-correct term. We are able to remove all $l$-unduplicated rules in such manner, so the lemma is proved.
\end{proof}

\begin{definition}
A path in the derivation tree is called $l$-heavy if all the nodes on this path are of rank $l$ or greater.
\end{definition}

\begin{definition}
A $k$-DCFL $G'$ is $m$-compact if for every word $w$ there is a derivation tree $T_w$ such that for every node $v$ of positive rank $l$ in $T_w$ there is an element $v'$ of $l'$-matryoshka for some $l' \geq l$, such that $v$ and $v'$ are connected by an $l$-heavy path whose length is not greater than $m$.
\end{definition}

\begin{theorem}\label{modify-rules}
For every $k$-DCFL $G$ there is an equivalent $k$-DCFL $G'$, which is $m$-compact.
\end{theorem}
\begin{proof}
See Appendix \ref{app-rules-modify}.
\end{proof}

\section{Main results}

In this section we use Theorem \ref{modify-rules} to prove a strengthed version of pumping lemma and an analogue of Ogden lemma for $k$-DCFGs.

\begin{definition}\footnote{Actually our notion of $l$-pump reformulates the definition of an even $k$-pump from \cite{Kanazawa2009}.}
We call a $l$-pump a pair of internal nodes $v$ and $v'$ of a derivation tree, such that $v$ and $v'$ has the same label of rank $l-1$ and $v'$ is the direct descendant of $v$. In this case $v$ is the top and $v'$ --- the bottom node of the pump.
\end{definition} 

\begin{theorem}
For any $k$-DCFL $L$ there is number $n$, such that any word $w \in L$ with $|w| > n$ can be represented in the form $w = s_0 y_1 u_1 z_1 s_1 y_2 u_2 z_2 s_2 \ldots y_l u_l z_l s_l$ for some $l \leq k + 1$, satisfying the following requirements:
\begin{enumerate}
\item $|y_1 z_1 \ldots y_l z_l | > 0$,
\item $|y_1 u_1 z_1 \ldots y_l u_l z_l| \leq n$,
\item For any $p \in \bbbn$ the word $s_0 y_1^p u_1 z_1^p s_1 \ldots y_l^p u_l z_l^p s_l$ belongs to $L(G)$.
\end{enumerate}
\end{theorem}
\begin{proof}
By theorem \ref{modify-rules} we assume that $L$ is generated by a $m$-compact grammar $G$ for some natural $m$. Let $N_l$ be the number of nonterminals of rank $l$ in this grammar. $N_+ = \max{(N_l | l > 0)}$ and $N = N_0 + N_+ + m$. We set $n = 2^N$.

Let $w \in L(G)$ be a word such that $|w| \geq n$ and $T_w$ be its derivation tree, deriving the term $\alpha$ and satisfying the requirement of Theorem \ref{modify-rules}, then $d(T_w) \geq n + 1$. Consider the $N_0 + 1$ deepest nodes of the longest branch of $T_w$. If all them are of rank $0$, then some pair of nodes have the same label and hence form a $1$-pump. If conversely, some node $v_t$ is of rank $t > 0$, then there is an element of some $l'$-matryoshka on the distance not greater than $m$ from $v_t$. Then the distance from $v_t$ to the upper node of this matryoshka is at most $m + N_+$.  
This $l'$-matryoshka contains an $l'+1$-pump, and the depth of the top node of this pump differs from the depth of $T_w$ by at most $N_0 + N_+ + m = n$. So we have proved an existence of such an $l$-pump for some $l \leq k+1$, that the depth of the subtree below its top node is at most $n$ (in this case $l = l'+1$).

Let $v$ and $v'$ be the top and bottom nodes of this pump, $B$ be their nonterminal label, and $C_1$ and $C_2$ be their outer contexts. Then $\alpha = C_1[C_2[\beta]]$ for some term $\beta$ satisfying the following properties:
\begin{enumerate}
\item $S \vdash C_1[B]$,
\item $B \vdash C_2[B]$,
\item $B \vdash \beta$.
\end{enumerate}

Let $\nu(\beta) = u_1 1 \ldots 1 u_l$. By Lemma \ref{direct-synt} the context $C_2[\gamma]$ is equivalent to $y_1 (\gamma \linebreak \otimes (z_1 1 y_2, \ldots, z_{l-1} 1 y_l)) z_l$ for some words $y_1, z_1, \ldots y_l, z_l \in \Sigma^*$ for any term $\gamma$ of rank $l-1$. Also $C_1[\eta] \sim s_0 (\eta \otimes (s_1, \ldots, s_{l-1})) s_l$ for some words $s_0, \ldots, s_l \in \Sigma^*$. Then $w$ is equivalent and hence equal to the word 
$s_0 ((y_1 ((u_1 1 \ldots 1 u_l) \otimes (z_1 1 y_2,\linebreak \ldots, z_{l-1} 1 y_l)) z_l) \otimes (s_1, \ldots, s_{l-1})) s_l = s_0 ((y_1 u_1 z_1 1 \ldots 1 y_l u_l z_l) \otimes (s_1, \ldots, s_{l-1})) s_l = s_0 y_1 u_1 z_1 s_1 \ldots y_l u_l z_l s_l$. The depth of $C_2[\beta]$ is not greater than $N$, so its value $y_1 u_1 z_1 1 \ldots 1 y_l u_l z_l$ cannot be longer than $n$. It remains to prove the third statement.

We denote by $C_2^p$ the context $C_2\!\!\!\!\!\!\underbrace{[\ldots[}_{\substack{(p-1)\text{ times}}} \!\!\!\!\!\!C_2] \ldots]$. Repeating the derivation $B \vdash C_2[B]$ $p$ times, we obtain the derivation $B \vdash C_2^p[B]$. Applying Lemma \ref{direct-synt} to the context $C_2$ several times and using Lemma \ref{term-eq-basic}, we get the equivalence $C_2^p[\gamma] \sim y_1^p (\gamma \otimes (z_1^p 1 y_2^p, \ldots, z_{l-1}^p 1 y_l^p)) z_l^p$. Setting $\gamma = \beta$ yields that $y_1^p u_1 z_1^p 1 \ldots 1 z_l^p u_l y_l^p \in L_{G'}(B)$ and consequently $s_0 y_1^p u_1 z_1^p s_1 \ldots z_l^p u_l y_l^p s_l \in L_{G'}(S)$.
The theorem is proved.
\end{proof}

Let the pair of nodes $v$ and $v'$ be an $l$-pump. We call its collapsing the replacement of subtree growing from $v$ by subtree growing from $v'$. The scope of an $l$-pump consists of the nodes being descendants of $v$ but not of $v'$; these are the nodes removed when collapsing this pump.

\begin{lemma}\label{collapse-pump}
Let $T'$ be a tree obtained from $T$ by collapsing some pump. If the nodes $v_1$ and $v_2$ form a pump in $T'$, then they have also formed a pump in $T$.
\end{lemma}
\begin{proof}
Let $v$ and $v'$ be, respectively, the top and bottom nodes of the collapsed pair. If $v'$ is not on the path from $v_1$ to $v_2$ in $T'$ then $v_2$ has already been a direct descendant of $v_1$ in $T$. Otherwise $(v_1, v')$ and $(v', v_2)$ are the pairs of direct descendants in $T'$, which means that$(v_1, v)$ and $(v', v_2)$ were the pairs of direct descendants in $T$. Using the fact that $v'$ was a direct descendant of $v$ in $T$ and the transitivity of direct descendance, we obtain that $v_2$ was a direct descendant of $v_1$ in $T$, implying they formed a pump. The lemma is proved.
\end{proof}

Lemma \ref{collapse-pump} implies that a terminal vertex being in scope of a pump in a collapsed derivation tree was also in scope of this pump in the original tree. This fact allows us to prove a weakened analogue of the Ogden lemma (\cite{Ogden1968}). 

\begin{theorem}[Ogden lemma for $1$-DCFGs]\label{Ogden-dcfg}
For any $k$-DCFL $L$ there is a number $t$ such that for any word $w \in L$ with at least $t$ selected positions there is a representation $w = s_0 y_1 u_1 z_1 s_1 \ldots y_l u_l z_l s_l$ for some $k \leq l + 1$, satisfying the following conditions:
\begin{enumerate}
\item For any $p \in \bbbn$ the word $s_0 y_1^p u_1 z_1^p \ldots y_l^p u_l z_l^p s_l$ belongs to $L(G)$.
\item There is at least one selected position in some of the words $y_1, z_1, \ldots, y_l, z_l$.
\end{enumerate}
\end{theorem}
\begin{proof}
We set $t$ equal to $n$ from pumping lemma. It suffices to show that one of the selected positions is in scope of some pump. We use an induction on $|w|$, note that this length is at least $n$. There is a presentation $w = x'_0 y'_1 u'_1 z'_1 x'_1 \ldots y'_l u'_l z'_l x'_l$ such that the word $x'_0 u'_l x'_1 \ldots u'_l x'_l$ is also in $L$. If the removed words contained a labeled position, the lemma is proved. Otherwise the word $w' = x'_0 u'_l x'_1 \ldots u'_l x'_l$ contains the same number of labeled positions and we can apply the induction hypothesis to its derivation tree $T'$, which is obtained from $T$ by collapsing. Then one of the selected positions is in scope of some pump in $T'$, which implies by Lemma \ref{collapse-pump} it was in scope of a pump in $T$ already. The theorem is proved.
\end{proof}

\section{Examples of non $1$-DCFLs}

In this section we use the established theoretical results to give some examples of non-$1$-DCFLs. To address this question we need to investigate more thoroughly the properties of constituents of displacement context-free grammars. A constituent is the fragment of the word derived from a node of derivation tree. In the context-free case every constituent is a continuous subword, hence it can be described by two numbers: the position of its first symbol and the position of its last symbol plus one (we add one to deal with empty constituents). Recall that context-free constituents must be correctly embedded which means they either do not intersect or one constituent is the part of the another. 

 The situation is a bit more complex in the case of DCFGs. However, the results of \cite{Sorokin2013STUS} provide analogous geometrical intuition. The constituents of rank $1$ are the words of the form $w_1 1 w_2$, where $w_1, w_2$ are continuous subwords of the derived word $w$. Then a constituent of rank $1$ is characterized by four indexes $i_1 \leq j_1 \leq i_2 \leq j_2$ of the borders of its subwords. We identify a constituent with a tuple of its characterizing indexes in the ascending order. The proofs of the statements below are carried out to the Appendix.

\begin{lemma}\label{const-variants}\footnote{Geometrical illustrations are also given in the Appendix.}
One of the possibilities below hold without loss of generality for any pair of constituents $(i_1, j_1, i_2, j_2)$ and $(i'_1, j'_1, i'_2, j'_2)$: 
\begin{enumerate}
\item $j_2 \leq i'_1$,
\item $j_1 \leq i'_1 \leq j'_2 \leq i_2$,
\item $i_1 \leq i'_1 \leq j'_2 \leq j_1$ or $i_2 \leq i'_1 \leq j'_2 \leq j_2$,
\item $i_1 \leq i'_1 \leq j'_1 \leq j_1 \leq i_2 \leq i'_2 \leq j'_2 \leq j_2$.
\end{enumerate}\end{lemma} 

Since every pump is just a pair of properly embedded constituents labeled by the same nonterminal, Lemma \ref{const-variants} helps to specify the mutual positions of different $2$-pumps. The scope of the pump contains exactly the positions which are in the top constituent but not in the bottom, so every $2$-pump is described by eight indexes $i_1 \leq j_1 \leq k_1 \leq l_1 \leq i_2 \leq j_2 \leq k_2 \leq l_2$, such that $(i_1, l_1, i_2, l_2)$ is the tuple of indexes of its top constituent and $(j_1, k_1, j_2, k_2)$ --- of the bottom.

\begin{lemma}\label{pumps-variants}\footnote{The illustrations are in the Appendix.}
One of the possibilities below hold without loss of generality for any pair of $2$-pumps $(i_1, j_1, k_1, l_1, i_2, j_2, k_2, l_2)$ and $(i'_1, j'_1, k'_1, l'_1, i'_2, j'_2, k'_2, l'_2)$: 
\begin{enumerate}
\item $l_2 \leq i'_1$,
\item $i_1 \leq i'_1 \leq l'_2 \leq j_1$ or $k_2 \leq i'_1 \leq l'_2 \leq l_2$,
\item $i_1 \leq\! i'_1 \leq\! j'_1 \leq\! j_1 \leq\! k_1 \leq\! k'_1 \leq\! l'_1 \leq\! l_1 \leq\! i_2 \leq\! i'_2 \leq\! j'_2 \leq\! j_2 \leq\! k_2 \leq\! k'_2 \leq\! l'_2 \leq\!l_2$,
\item $i_1 \leq\! i'_1 \leq\! j'_1 \leq\! k'_1 \leq\! j_1 \leq\! k_1 \leq\! l'_1 \leq\! l_1 \leq\! i_2 \leq\! i'_2  \leq\! j_2 \leq\! k_2 \leq\! j'_2 \leq\! k'_2 \leq\! l'_2 \leq\!l_2$,
\item $i_1 \leq\! i'_1 \leq\! j_1 \leq\! k_1 \leq\! j'_1 \leq\! k'_1 \leq\! l'_1 \leq\! l_1 \leq\! i_2 \leq\! i'_2 \leq\! j'_2 \leq\! k'_2  \leq\! j_2 \leq\! k_2 \leq\! l'_2 \leq\!l_2$,
\item $i_1 \leq\! i'_1 \leq\! j_1 \leq\! j'_1 \leq\! k'_1 \leq\! k_1 \leq\! l'_1 \leq\! l_1 \leq\! i_2 \leq\! i'_2 \leq\! j_2 \leq\! j'_2 \leq\! k'_2 \leq\! k_2 \leq\! l'_2 \leq\! l_2$,
\item $k_1 \leq i'_1 \leq l'_1 \leq l_1 \leq i_2 \leq i'_2 \leq l'_2 \leq j_2$,
\item $i_1 \leq i'_1 \leq l'_1 \leq j_1 \leq k_2 \leq i'_2 \leq l'_2 \leq l_2$,
\item $k_1 \leq i'_1 \leq l'_2 \leq l_1$ or $i_2 \leq i'_1 \leq l'_2 \leq j_2$,
\item $j_1 \leq i'_1 \leq l'_1 \leq k_1 \leq j_2 \leq i'_2 \leq l'_2 \leq k_2$,
\item $j_1 \leq i'_1 \leq l'_2 \leq k_1$ or $j_2 \leq i'_1 \leq l'_2 \leq k_2$,
\item $l_1 \leq i'_2 \leq l'_2 \leq i_2$.
\end{enumerate}\end{lemma}

Let $\pi_1 = (i_1, j_1, k_1, l_1, i_2, j_2, k_2, l_2)$ and $\pi_2 = (i'_1, j'_1, k'_1, l'_1, i'_2, j'_2, k'_2, l'_2)$ be two $2$-pumps. We call a pair of $\pi_1$ and $\pi_2$ linear if $l_2 \leq i'_1$ or $l'_2 \leq i_1$. We call $\pi_1$ outer for the pump $\pi_2$ if $i_1 \leq i'_1 \leq l'_2 \leq l_2$. Note that if a pair of $2$-pumps is not linear, then one of its elements is the outer pump for another. We call $\pi_1$ embracing for $\pi_2$ if $l_1 \leq i'_1 \leq l'_2 \leq i_2$.

\begin{corollary}\label{pumps-geom}
Let $(i_1, j_1, k_1, l_1, i_2, j_2, k_2, l_2)$ and $(i'_1, j'_1, k'_1, l'_1, i'_2, j'_2, k'_2, l'_2)$ be $2$-pumps such that one of the segments of the second pump is a proper subset of the segment $[l_1; i_2]$. Then either the second pump is outer for the first (which means $i'_1 \leq i_1 \leq l_2 \leq l'_2$) or the first pump is embracing for the second.
\end{corollary}

Lemma \ref{pumps-variants} allows us to give some examples of non $1$-DCFLs. The first example is the language $4\mathrm{MIX} = \{w \in \{a, b, c, d\}^* \mid |w|_a = |w|_b = |w|_c = |w|_d\}$.

\begin{theorem}\label{4-MIX}
The language $4$MIX cannot be generated by any $1$-DCFG.
\end{theorem}
\begin{proof}
Since wMCFLs are closed under intersection with regular languages, it suffices to prove that the language $4\mathrm{MIX} \cap (a^+b^+c^+d^+)^2$ is not a $1$-DCFL. Assume the contrary, let $t$ be the number from Ogden's lemma applied to this language. Let the word $w = a^{m_1}b^{m_2}c^{m_3}d^{m_4} a^{n_1}b^{n_2}c^{n_3}d^{n_4}$ satisfy the following conditions:
\begin{enumerate}
\item $\min{(m_j, n_j)} \geq t$,
\item $m_1 \geq (3M + 1)(M + t)$, where $M = \max{(m_2, m_4, n_3)}$,
\item $m_4 \geq (n_1 + 1)(n_1 + t)$.
\end{enumerate}

Note that every $2$-pump contains an equal number of $a$-s, $b$-s, $c$-s and $d$-s, and every continuous segment of it consists of identical symbols (we call such segments homogeneous). We enumerate the maximal continuous homogeneous subwords of $w$ from $1$ to $8$. Then every $2$-pump intersects with exactly $4$ of such segments. We call a $[d_1,\ldots, d_l]$-pump a pump intersecting the segments with numbers $d_1, \ldots, d_l$ (and possibly some others).

We select $3M + 1$ segments of length $t$ in the first segment of the word $w$ so, that any two segments are separated by not less than $M$ symbols. By Theorem \ref{Ogden-dcfg} each such segment intersects with some $2$-pump. We want to prove that some of them intersects with a $[1,3,6,8]$-pump. Indeed, any two points from different segments cannot belong to the same $[1,7]$-pump since in this case there is a continuous segment of at least $M+1$ $a$-s in the pump, then the pump contains at least $M+1$ $c$-s, which exceeds the length of the $7$-th segment. It follows that there are at most $M$ $[1,7]$-pumps, by the same arguments there are at most $M$ $[1,2]$-pumps and at most $M$ $[1,4]$-pumps, therefore the number of $[1]$-pumps which are not $[1,3,6,8]$-pumps is less than $3M+1$ which proves the existence of a $[1,3,6,8]$-pump.

By the same arguments there is at least one $[4]$-pump $\pi$, which is not a $[4,5]$-pump. By corollary \ref{pumps-geom} applied to the $[1,3,6,8]$-pump $\pi'$, either $\pi$ is an outer pump for $\pi'$ (which means $\pi$ should be a $[1,4,8]$-pump) or $\pi$ is embraced by $\pi'$. In the first case there are two $d$-segments in the pump, in the second case it should be a $[3,4,5,6]$-pump which contradicts our assumption. So we have reached a contradiction and the theorem is proved. 
\end{proof}
    
Our technique of embedding different $2$-pumps also works in a more complex case. Consider the language $MIX = \{ w \in \{a, b, c\}^* \mid |w|_a = |w|_b = |w|_c \}$. It is expected to be not a $DCFL$ since it demonstrates an extreme degree of unprojectivity. It is proved in \cite{KanazawaSalvati2012} that $MIX$ is not a $2$-wMCFL (and hence not a $1$-DCFL). The proof extensively uses geometrical arguments and is therefore very difficult to be generalized for similar languages or wMCFGs of higher order. Our proof uses only the Ogden's lemma for DCFGs and is much shorter.

\begin{theorem}
The $MIX$ language is not a $1$-DCFL.
\end{theorem}
\begin{proof}
We use the same method and notation as in the case of $4$MIX language. Again, it suffices to prove that the language $L = MIX \cap a^+b^+c^+b^+c^+a^+$ is not a $1$-DCFL. Let $t$ be the number from Ogden's lemma for $L$. Consider the word $w = a^{m_1}b^{m_2}c^{m_3}b^{n_2}c^{n_3}a^{n_1}$ satisfying the following properties:
\begin{enumerate}
\item $\min{(m_j, n_j)} \geq t$,
\item $m_1 \geq (2M + 1)(M + t)$, where $M = \max{(m_3, n_2)}$,
\item $n_1 \geq (2M + 1)(M + t)$, where $M = \max{(m_3, n_2)}$,
\item $m_3 \geq (n_2 + 1)(n_2 + t)$.
\end{enumerate}

By the same arguments as in Theorem \ref{4-MIX} we establish the existence of $[1,2,5]$- and $[2,5,6]$-pumps. Since they cannot form a linear pair of $2$-pumps, one of them is an outer pump for another, which implies one of them is a $[1,2,5,6]$-pump, we denote this pump by $\pi'$. The condition $m_3 \geq (n_2 + 1)(n_2 + t)$ implies the existence of a $[2,3]$-pump $\pi$. If $\pi$ is embraced by a $[1,2,5,6]$-pump, then it contains no $a$-s, which is impossible, therefore by Corollary \ref{pumps-geom} $\pi$ is an outer pump $\pi'$ and actually a $[1,2,3,6]$-pump.

The condition $n_2 \geq t$ implies the existence of a $[4]$-pump, which is a $[1,4,6]$-pump by the same arguments as in the previous paragraph. To be correctly embedded with the $[1,2,3,6]$-pump it should be a $[1,2,4,6]$-pump but there are no $c$-s in such pump. Hence we reached the contradiction and the $MIX$-language cannot be generated by a $1$-DCFG. The theorem is proved.
\end{proof} 

\section{Conclusions and future work}

We have proved a strong version of the pumping lemma and a weak Ogden lemma for the class of DCFLs which is also the class of well-nested multiple context-free languages. These statements allow us to prove that some languages, like the well-known MIX-language, do not belong to the family of $1$-DCFLs or, in other terms, the family of tree adjoining languages. We hope to adopt the proof for the case of semiblind three-counter language $\{ w \in \{a, b, c\}^* \mid |w|_a = |w|_b = |w|_c, \forall u \sqsubseteq w \: |u|_a \geq |u|_b \geq |u|_c \}$ to prove that a shuffle iteration of a one-word language may lie outside the family of $1$-DCFLs. The author supposes that the technique used in the article will work not also in the case of $2$-pumps, but also in a more complex cases. We hope that our results will help understand better the structure of well-nested MCFLs and, in particular, prove the Kanazawa-Salvati conjecture, which states that MIX is not a well-nested MCFL.

\section{Acknowledgements}

The author thanks Makoto Kanazawa for his helpful suggestions and the anonymous referees of DLT 2014 conference, whose thoughtful comments essentially improved the paper.

\bibliographystyle{splncs_srt}
\bibliography{../../Bibliography/algebra,../../Bibliography/grammars,../../Bibliography/sorokin,../../Bibliography/various,../../Bibliography/books,../../Bibliography/discontinuous_Lambek}

\appendix
\section{Multicontext equivalence}\label{app-term-eq}

A multicontext $C$ is called $k$-correct, if the rank of all its submulticontexts does not exceed $k$. If only the rank of $C$ as well as rank of all its leafs is not greater than $k$, the multicontext is called $k$-essential.

\begin{lemma}
For any $k$-essential multicontext $C$ there is an equivalent $k$-correct multicontext $C'$.
\end{lemma}
\begin{proof}
At first we prove that there is an equivalent multicontext with no internal submulticontexts of rank greater than $k$. Let $K$ be the maximal rank of submulticontexts in $C$, a submulticontext occurrence is called heavy if the rank of the corresponding submulticontext equals $K$. We use induction on $K$ and the number of heavy submulticontext occurrences.

Let $(C_1, v_1)$ be an occurrence a submulticontext $C_1$ of rank $K$ with minimal depth among all such occurrences. $v$ cannot be the root of $tree(C)$ since $rk(C) < K$ so let $(C_2, v_2)$ be the submulticontext occurrence corresponding to the parent of $v_1$ in the syntactic tree. Then $\rk(C_2) < rk(C_1)$ which implies that $C_2 = C_1 \odot_j E$ for some $0$-ranked multicontext $E$. We transform $C_2$ to an equivalent submulticontext $C'_2$ with less occurrences of heavy submulticontexts. The transformation uses the equivalences from Lemma \ref{multicontext-eq}.

Consider the possible structure of the multicontext $C_1$. First, let it has the form $C_1 = C_3 \cdot C_4$. If $rk(C_3) \leq j$ then the multicontext $C'_2 = (C_3 \odot_j E) \cdot C_4$ is equivalent to $C_2$ and has less occurrences of submulticontexts of rank $K$ because we have removed the occurrence of submulticontext $C_1$ and haven't add any other heavy multicontexts. In case $j > \rk(\beta)$ the multicontext $C'_2 = C_3 \cdot (C_4 \odot_{j-\rk(C_3)} E)$ does the same job.

Now let $C_1$ have the form  $C_1 = C_3 \odot_l C_4$. If $j < l$ then we define $C'_2 = (C_3 \odot_j E) \odot_{l-1} C_4$. In case $l \leq j < l + \rk(C_4)$ we set $C'_2 = C_3 \odot_l (C_4 \odot_{j-l+l} E)$ and in case $j \geq l + \rk(\beta)$ we define $C'_2 = (C_3 \odot_{l+\rk(C_4) - 1} E) \odot_j C_4$. In all the cases $C'_2$ is equivalent to $C_2$ by lemma \ref{multicontext-eq} and has fewer occurrences of heavy submulticontexts.

Since $C_2$ was a submulticontext of $C$, there is a context $C_0$ such that $C = C_0[C_2]$. Then the multicontext $C'' = C_0[C'_2]$ is equivalent to $C$ and has fewer occurrences of heavy submulticontexts. We can apply the induction hypothesis to $C''$ and obtain the required multicontext $C'$.

In case $k \geq 1$ the proof is completed since the rank of elements of $\Sigma_1$ cannot be greater then $1$ and the ranks of other atomic submulticontexts are smaller than $k$ by theorem conditions. In the case $k=0$ we need a minor complication since some leaves of $tree(C')$ might be labeled by $1$. However, they all occur in submulticontexts of the form $1 \odot_1 E$ for some $E$ of rank $0$ since $C$ has no internal submulticontexts of positive rank. If we replace all submulticontexts of the form $1 \odot_1 E$ by the corresponding multicontext $E$, we obtain an equivalent $0$-correct multicontext. The lemma is proved.
\end{proof}

\section{Modifying derivation trees}\label{app-rules-modify}

\begin{definition}
A $k$-DCFL $G'$ is $m$-compact if for every word $w$ there is a derivation tree $T_w$ such that for every node $v$ of positive rank $l$ in $T_w$ there is an element $v'$ of $l'$-matryoshka for some $l' \geq l$, such that the length of the path between $v$ and $v'$ is not greater than $m$ and all the nodes in this path has rank $l$ or greater.
\end{definition}

\begin{theorem}\label{modify-rules-proof}
For every $k$-DCFL $G$ there is an equivalent $k$-DCFL $G'$, which is $m$-compact.
\end{theorem}
\begin{proof}
We inductively construct the grammars $G_{k+1} = G, G_k, \ldots, G_1 = G'$ satisfying the following properties: 
\begin{enumerate}
\item All these grammars are equivalent.
\item For any grammar $G_l$ there exists a natural number $m_l$ such that for any word $w \in L(G_l)$ there is a derivation tree $T_{w,l}$ of the grammar $G_l$, satisfying the following properties: for any node $v'$ of rank $l' \geq l$ in this tree there is an element $v''$ of some $l''$-matreshka connected with $v'$ by an $l$-heavy path, whose length is not greater than $m_l$. 
\end{enumerate}

We set $m_{k+1} = 0$ and take $m = m_1$. For $G_{k+1}$ the conditions specified are trivial since there are no nonterminals of rank $k+1$ or greater. The grammar $G_l$ is obtained from $G_{l+1}$ by duplicating all the derivable rules in $G_l$, whose right sides are $l$-redundant. The duplication process is the same as in Lemma \ref{l-duplication}. We also set $m_l = m_{l+1} + 2N_l$ where $N_l$ is the number of nonterminals of rank $l$ in $G_{l}$ (or $G_{l+1}$ since these numbers are equal by Lemma \ref{l-duplication}). The lemmas below justify that all the constructed grammars satisfy the desirable properties.
\end{proof}

In the further we fix some number $l$ and consider only the derivations in the grammar $G_l$.

\begin{lemma}\label{compact-old-lemma}
For any $w \in L(G)$ there is a $S$-derivable term $\alpha$ such that $w = \mu(\alpha)$ with the derivation $\deriv{D_l}{S}{\alpha}$ whose tree $T_{D,l}$ satisfy the following properties:
\begin{enumerate}
\item The vicinity of any node $v$ of rank $l$ either contains a node of greater rank or an element of some $l$-matreshka.
\item Any node $v'$ of rank $l' > l$ is connected with some element $v''$ of some $l''$ -matreshka with $l'' \geq l'$ by an $l'$-heavy path whose length does not exceed $m_{l+1}$.
\end{enumerate}
\end{lemma}
\begin{proof}
We start from the derivation tree $T_{D, l+1}$ of $w$ in the grammar $G_{l+1}$. This tree remains a correct derivation tree in $G_l$ and satisfies the second property by induction hypothesis. We want to reduce the number of nodes of rank $l$ whose vicinity does not satisfy the first statement of the lemma, preserving the second property.
Let $N_l$ be the number of nonterminals of rank $l$ in the grammar $G_l$.

Let $T$ be a tree obtained on some stage of this process and $v$ be its node of rank $l$ whose vicinity $U_v$ violates the first property. Consequently, $U_v$ contains no subbranches with more than $N_l$ consecutive nodes of rank $l$, which means that $d(U_v) \leq N_l + 1$ since all internal nodes of $U_v$ has rank $l$. Additionally, since the second part of the first property is incorrect, the rank of all $U_v$ leaves is less than $l$, as well is the rank of its root. Let $B$ denote the nonterminal label of the root of $U_v$. 

Hence, the term $\beta = C_{\beta}[B_1, \ldots, B_t]$ corresponding to $U_v$ is $l$-essential ($C$ denotes the skeleton of $B$). By definition, the rule $B \to \beta$ is $l$-redundant in $G_l$; since $G_l$ has the same nonterminals of rank $l$ as $G_{l+1}$, this rule was also redundant in $G_{l+1}$. Therefore by construction $G_l$ has an equivalent derivable rule $B \to \beta'$ for some $(l-1)$-correct term $\beta'$. Recall that $\beta' = C_{\beta'}[B_1, \ldots, B_t]$ for some ground multicontext $C_{\beta'}$ equivalent to $C_{\beta}$.

Let $\alpha$ denote the term derivable by the tree $T$. Then $\alpha = C_0[C_{\beta}[\beta_1, \ldots, \beta_t]]$ for some terms $\beta_1, \ldots, \beta_t$ derivable from $B_1, \ldots, B_t$ respectively. Consider the term $\alpha' = C_0[C_{\beta'}[\beta_1, \ldots, \beta_t]]$, it is also derivable from $A$ since the rule $B \to \beta'$ is derivable in $G_l$ by construction and the remaining derivation is the same. It also derives the word $w$, because term equivalence is a congruence relation. Since $\beta'$ is $(l-1)$-essential, we have removed a node with incorrect vicinity, so it remains to show that the second property is preserved.

Let $T'$ be the derivation tree of $\alpha'$, it is obtained from $T$ by replacing the vicinity $U_v$ with the derivation tree of $\beta'$ from $B$. Consider the $l'$-heavy path connecting some node $v'$ of rank $l' > l$ with an element $v''$ in some matreshka of rank $l'' \geq l'$ in the old tree $T$. This path cannot traverse $U_v$ since all the leaf nodes of $U_v$ are of smaller rank. Since $G_l$ contains the same number of nonterminals of rank $l$ and greater as $G_{l+1}$ does, $v''$ remains an element of an $l''$-matreshka in $T'$. Choosing the same $l'$-heavy path as in $T$, we provide the second property. Repeating the described procedure, we also provide the first property, so the lemma is proved.
\end{proof}

\begin{lemma}
The tree $T_{D,l}$ constructed in the previous lemma also satisfies the following property: for any node $v$ of rank $l$ it is connected with some element $v'$ of an $l'$-matreshka by an $l$-heavy path, whose length is at most $m_{l+1}$.
\end{lemma}
\begin{proof}
For every node $v$ of rank $l$ we consider its vicinity $U_v$. There are two possibilities: the depth of $v$ in $U_v$ is greater than $N_l$ and it is at most $N_l$. In the first case there is subbranch of length at least $N_l+1$ which contains $v$ and consists only of internal nodes of $U_v$. Hence $v$ is the element of an $l$-matreshka itself and satisfies the requirements of the lemma. 

Now the depth of $v$ is not greater than $N_l$. If the root of $U_v$, which we denote by $v_0$, is of rank greater than $l$, then the distance between $v$ and $v_0$ is at most $N_l$ and all the intermediate nodes are of rank $l$. Extending this path by a sequence of nodes with rank $l+1$ and greater from $v_0$ to its closest element of $l'$-matreshka with $l' > l$, we obtain the path from $v$ to the same matreshka element. Note that its length is not greater then $m_{l+1} + N_l \leq m_l$.

If the rank of $v_0$ is less then $l$ then it cannot have two children of rank $l$, but has only one such child $v_1$. All other nodes of rank $l$ in $U_v$ are direct descendants of $v_1$. If $d(U_v) \geq N_{l} + 2$, then $v_1$ is an element of an $l$-matreshka. Since the number of nodes between $v_1$ and $v$ is at most $N_l < m_l$, the requirements of lemma are again satisfied.

The only remaining case is when $d(U_v) < N_{l} + 2$ and $U_v$ contains a node of rank $l+1$ or greater. Let $v_2$ be such a node and $l'$ be its rank. Then $v_2$ is a descendant of $v_1$ and the distance between these nodes is at most $N_l + 1$. Therefore the path between $v$ and $v_2$ consists of at most $2N_l$ edges and intermediate nodes are of rank $l$. By Lemma \ref{compact-old-lemma} $v_2$ is connected by an $l'$-heavy path of length at most $m_{l+1}$ with an element $v''$ of an $l''$-matreshka for some $l'' > l'$. Then the distance between $v$ and $v''$ is not greater than $m_{l+1} + 2N_l = m_l$ which proves the final case. The lemma is proved.
\end{proof}

These two lemmas imply Theorem \ref{modify-rules-proof}. 

\section{Constituents in displacements context-free grammars}

This section we discuss the geometrical interpretation of constituents in displacement context-free grammars. A constituent is a (possibly discontinuous) fragment of a word derived from a node of its derivation tree. The nonterminal label of this node is the label of the constituent. In the basic context-free case the constituents are just continuous subwords, so every constituent is completely defined by two indexes $i, j$: the position of its first symbol and the position of its last symbol plus one (we add one to deal with empty constituents). Different constituents should satisfy the embedding conditions: either one of them is inside the other ($[i; j] \subseteq [i'; j']$ or $[i'; j'] \subseteq [i; j]$ in terms of indexes), or they do not have common internal points ($[i; j] \cap [i'; j']$ is one of the sets $\emptyset, \{i\}, \{j\}$). Two principal variants for mutual positions of different constituents are shown on the picture below.

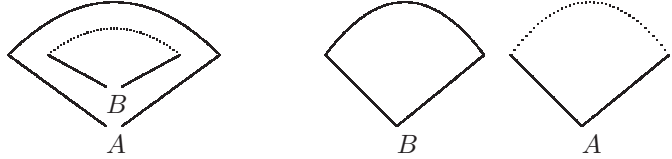
\begin{center}
\scalebox{1.0}{
\begin{picture}(270,85)
\qbezier(5,50)(45,90)(85,50)\qbezier[40](20,50)(45,70)(70,50)
\put(42, 28){$B$}\put(42, 13){$A$}
\qbezier(5,50)(5,50)(42,23)\qbezier(85,50)(85,50)(48,23)
\qbezier(20,50)(20,50)(42,38)\qbezier(70,50)(70,50)(48,38)
\qbezier(125,50)(155,90)(185,50)\qbezier[40](195,50)(225,90)(255,50)
\put(152, 13){$B$}\put(222, 13){$A$}
\qbezier(125,50)(125,50)(152,23)\qbezier(185,50)(185,50)(152,23)
\qbezier(195,50)(195,50)(222,23)\qbezier(255,50)(255,50)(222,23)
\end{picture}}
\end{center}

Let us now inspect the constituent structure of $1$-DCFGs. In the case of these grammars every constituent is either a continuous subword, if its label is of rank $0$, or a word of the form $w_1 1 w_2$ where $w_1$ and $w_2$ are continuous segments of the derived word $w$, if the label is of rank $1$. We focus our attention on the latter case because nothing has changed from the context-free case for the constituents of rank $0$. Then the first continuous part of the constituents is described by indexes $i_1, j_1$ and the second part by indexes $i_2, j_2$. Therefore every constituent of rank $1$ corresponds to a tuple $(i_1, j_1, i_2, j_2)$ of its indexes taken in the ascending order. We will not distinguish constituents and their index tuple in the further.

The following lemma about mutual positions of different constituents was proved in \cite{Sorokin2013STUS} in a more general case.

\begin{lemma}\label{const-variants-proof}
One of the possibilities below hold without loss of generality for any pair of constituents $(i_1, j_1, i_2, j_2)$ and $(i'_1, j'_1, i'_2, j'_2)$: 
\begin{enumerate}
\item $j_2 \leq i'_1$,
\item $j_1 \leq i'_1 \leq j'_2 \leq i_2$,
\item $i_1 \leq i'_1 \leq j'_2 \leq j_1$ or $i_2 \leq i'_1 \leq j'_2 \leq j_2$,
\item $i_1 \leq i'_1 \leq j'_1 \leq j_1 \leq i_2 \leq i'_2 \leq j'_2 \leq j_2$.
\end{enumerate}\end{lemma} 

The statement of the lemma above has a nice geometrical interpretation. We associate with every constituent $(i_1, j_1, i_2, j_2)$ of rank $1$ the following curve (the constituents themselves are marked by horizontal lines):

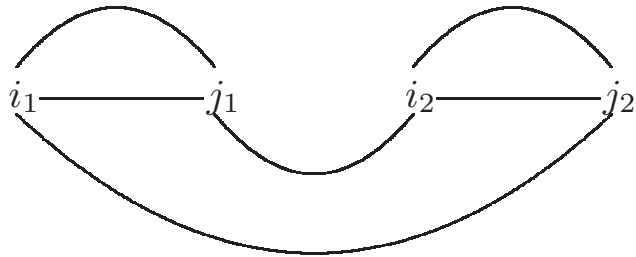
\scalebox{1.5}{
\begin{picture}(200,80)
{
\put(25,50){$i_1$}\put(75,50){$j_1$}\put(125,50){$i_2$}\put(175,50){$j_2$}
\qbezier(27,60)(52,90)(77,60)\qbezier(127,60)(152,90)(177,60)
\qbezier(77,48)(102,18)(127,48)\qbezier(27,48)(102,-22)(177,48)
\qbezier(33,52)(33,52)(74,52)\qbezier(133,52)(133,52)(174,52)
}
\end{picture}
}

The remarkable property of this interpretation is that if we write a derived word of the abscissa axis, enumerate the positions in it and draw the curves corresponding to all its constituents, then these curves must not intersect except the limit points. On the picture below we show all principal variants of different constituents location (solid and dash horizontal lines mark the constituents themselves). 

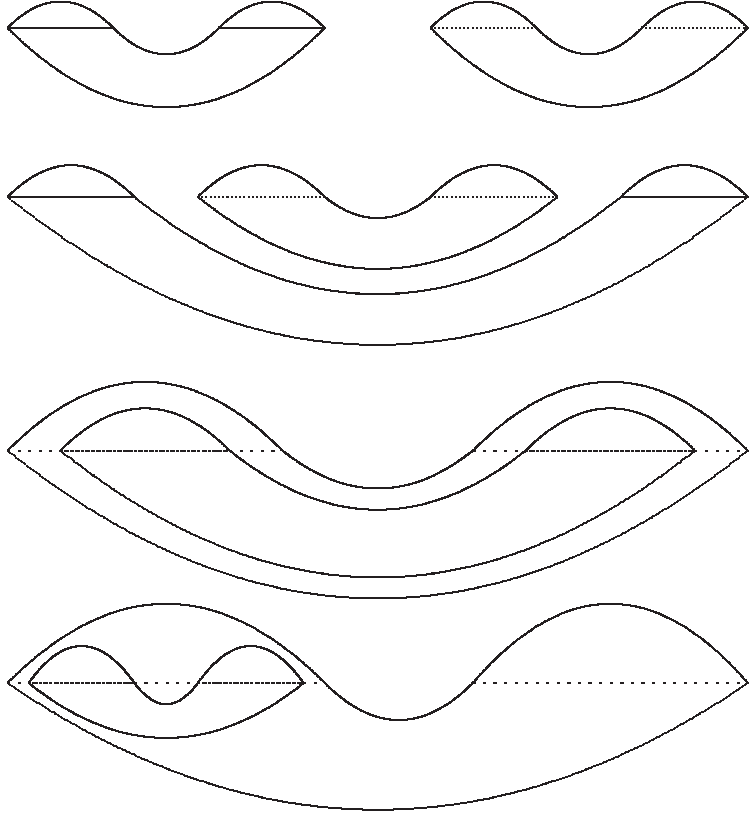
\scalebox{0.8}{
\begin{picture}(400,430)
\qbezier(10,410)(35,435)(60,410)\qbezier(110,410)(135,435)(160,410)
\qbezier(10,410)(10,410)(60,410)\qbezier(110,410)(135,410)(160,410)
\qbezier(60,410)(85,385)(110,410)\qbezier(10,410)(85,335)(160,410)
\qbezier(210,410)(235,435)(260,410)\qbezier(310,410)(335,435)(360,410)
\qbezier[25](210,410)(235,410)(260,410)\qbezier[25](310,410)(335,410)(360,410)
\qbezier(260,410)(285,385)(310,410)\qbezier(210,410)(285,335)(360,410)

\qbezier(10,330)(40,360)(70,330)\qbezier(300,330)(330,360)(360,330)
\qbezier(10,330)(40,330)(70,330)\qbezier(300,330)(330,330)(360,330)
\qbezier(10,330)(185,190)(360,330)\qbezier(70,330)(185,238)(300,330)
\qbezier(100,330)(130,360)(160,330)\qbezier(210,330)(240,360)(270,330)
\qbezier[30](100,330)(130,330)(160,330)\qbezier[30](210,330)(240,330)(270,330)
\qbezier(100,330)(185,262)(270,330)\qbezier(160,330)(185,310)(210,330)

\qbezier(10,210)(75,275)(140,210)\qbezier(230,210)(295,275)(360,210)
\multiput(10, 210)(5,0){26}{\line(1,0){1}}
\multiput(230, 210)(5,0){26}{\line(1,0){1}}
\qbezier(10,210)(185,70)(360,210)\qbezier(140,210)(185,174)(230,210)
\qbezier(35,210)(75,250)(115,210)\qbezier(255,210)(295,250)(335,210)
\multiput(37, 210)(5,0){16}{\line(1,0){2}}
\multiput(257, 210)(5,0){16}{\line(1,0){2}}
\qbezier(35,210)(185,90)(335,210)\qbezier(115,210)(185,154)(255,210)

\qbezier(10,100)(85,175)(160,100)\qbezier(230,100)(295,175)(360,100)
\multiput(10, 100)(5,0){30}{\line(1,0){1}}
\multiput(230, 100)(5,0){26}{\line(1,0){1}}
\qbezier(10,100)(185,-20)(360,100)\qbezier(160,100)(195,65)(230,100)
\qbezier(20,100)(45,135)(70,100)\qbezier(100,100)(125,135)(150,100)
\multiput(22, 100)(5,0){10}{\line(1,0){2}}
\multiput(102, 100)(5,0){10}{\line(1,0){2}}
\qbezier(20,100)(85,48)(150,100)\qbezier(70,100)(85,80)(100,100)

\end{picture}
}

Provided geometrical interpretation is very helpful in our main task: studying mutual postions of different pumps. Indeed, every pump is defined by its top and bottom nodes, which carry the same nonterminal labels and are connected by the path of nodes of the same rank. Since every node of the derivation tree corresponds to a constituent, then a pump is matched with a pair of embedded constituents with the same label. As earlier, we concentrate on the $4$-pumps which correspond to a pair of constituents of rank $1$. Then every pump can be defined by $8$ numbers $i_1 \leq j_1 \leq k_1 \leq l_1 \leq i_2 \leq j_2 \leq k_2 \leq l_2$ such that $(i_1, l_1, i_2, l_2)$ are the indexes of its top consituent and $(j_1, k_1, j_2, k_2)$ --- of the bottom. We call the segments $[i_1; j_1]$, $[k_1; l_1]$, $[i_2; j_2]$, $[k_2; l_2]$ the segments of the pump and identify a pump with the ascending tuple of its indexes. Below we illustrate how two constituents of rank $1$ with the same label form a $4$-pump:

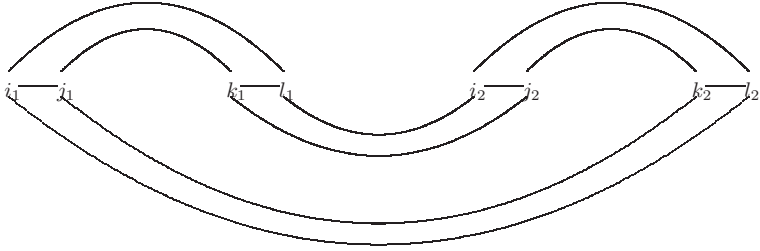
\begin{center}
\scalebox{0.8}{
\begin{picture}(400,150)
\put(8,98){$i_1$}\put(33,98){$j_1$}\put(113,98){$k_1$}\put(138,98){$l_1$}
\put(228,98){$i_2$}\put(253,98){$j_2$}\put(333,98){$k_2$}\put(358,98){$l_2$}
\qbezier(15,103)(15,103)(33,103)
\qbezier(120,103)(120,103)(138,103)
\qbezier(235,103)(235,103)(253,103)
\qbezier(340,103)(340,103)(358,103)
\qbezier(10,110)(75,175)(140,110)\qbezier(230,110)(295,175)(360,110)
\qbezier(10,98)(185,-42)(360,98)\qbezier(140,98)(185,62)(230,98)
\qbezier(35,110)(75,150)(115,110)\qbezier(255,110)(295,150)(335,110)
\qbezier(35,98)(185,-22)(335,98)\qbezier(115,98)(185,42)(255,98)
\end{picture}
}
\end{center}

Since the curves on the picture are the bounding curves for the constituents forming the pump, the curves corresponding to different pumps must not intersect anywhere except the abscissa axis. The following lemma interprets the geometrical conditions on correct embedding in terms of pump segments:

\begin{lemma}\label{pumps-variants-proof}
One of the possibilities below hold without loss of generality for any pair of $4$-pumps $(i_1, j_1, k_1, l_1, i_2, j_2, k_2, l_2)$ and $(i'_1, j'_1, k'_1, l'_1, i'_2, j'_2, k'_2, l'_2)$: 
\begin{enumerate}
\item $l_2 \leq i'_1$,
\item $i_1 \leq i'_1 \leq l'_2 \leq j_1$ or $k_2 \leq i'_1 \leq l'_2 \leq l_2$,
\item $i_1 \leq\! i'_1 \leq\! j'_1 \leq\! j_1 \leq\! k_1 \leq\! k'_1 \leq\! l'_1 \leq\! l_1 \leq\! i_2 \leq\! i'_2 \leq\! j'_2 \leq\! j_2 \leq\! k_2 \leq\! k'_2 \leq\! l'_2 \leq\!l_2$,
\item $i_1 \leq\! i'_1 \leq\! j'_1 \leq\! k'_1 \leq\! j_1 \leq\! k_1 \leq\! l'_1 \leq\! l_1 \leq\! i_2 \leq\! i'_2  \leq\! j_2 \leq\! k_2 \leq\! j'_2 \leq\! k'_2 \leq\! l'_2 \leq\!l_2$,
\item $i_1 \leq\! i'_1 \leq\! j_1 \leq\! k_1 \leq\! j'_1 \leq\! k'_1 \leq\! l'_1 \leq\! l_1 \leq\! i_2 \leq\! i'_2 \leq\! j'_2 \leq\! k'_2  \leq\! j_2 \leq\! k_2 \leq\! l'_2 \leq\!l_2$,
\item $i_1 \leq\! i'_1 \leq\! j_1 \leq\! j'_1 \leq\! k'_1 \leq\! k_1 \leq\! l'_1 \leq\! l_1 \leq\! i_2 \leq\! i'_2 \leq\! j_2 \leq\! j'_2 \leq\! k'_2 \leq\! k_2 \leq\! l'_2 \leq\! l_2$,
\item $k_1 \leq i'_1 \leq l'_1 \leq l_1 \leq i_2 \leq i'_2 \leq l'_2 \leq j_2$,
\item $i_1 \leq i'_1 \leq l'_1 \leq j_1 \leq k_2 \leq i'_2 \leq l'_2 \leq l_2$,
\item $k_1 \leq i'_1 \leq l'_2 \leq l_1$ or $i_2 \leq i'_1 \leq l'_2 \leq j_2$,
\item $j_1 \leq i'_1 \leq l'_1 \leq k_1 \leq j_2 \leq i'_2 \leq l'_2 \leq k_2$,
\item $j_1 \leq i'_1 \leq l'_2 \leq k_1$ or $j_2 \leq i'_1 \leq l'_2 \leq k_2$,
\item $l_1 \leq i'_2 \leq l'_2 \leq i_2$.
\end{enumerate}\end{lemma}
\begin{proof}
We derive the current lemma formally from Lemma \ref{const-variants-proof}, illustrating the proof by geometrical arguments. We call a pair of $4$-pumps linear if $l_2 \leq i'_1$ or $l'_2 \leq i_1$. We call the pump $(i_1, j_1, k_1, l_1, i_2, j_2, k_2, l_2)$ outer for the pump $(i'_1, j'_1, k'_1, l'_1, i'_2, j'_2, k'_2, l'_2)$ if the condition $i_1 \leq i'_1 \leq l'_2 \leq l_2$ holds. Note that if two pumps do not form a linear pair, then one of them is an outer for the other.

We denote $\pi = (i_1, j_1, k_1, l_1, i_2, j_2, k_2, l_2)$ and $\pi' = (i'_1, j'_1, k'_1, l'_1, i'_2, j'_2, k'_2, l'_2)$ to shorten the notation. If the pair of $\pi$ and $\pi'$ is linear then up to renaming the pumps the first alternative of the lemma holds. Otherwise one of the pumps is the outer for another, let $\pi$ be such a pump. So $i_1 \leq i'_1 \leq l'_2 \leq l_2$. Consider the constituents $(i_1, l_1, i_2, l_2)$, $(j_1, k_1, j_2, k_2)$, $(i'_1, l'_1, i'_2, l'_2)$, $(j'_1, k'_1, j'_2, k'_2)$, each of then bounds a region on the plane. By the geometric interpretation of Lemma \ref{const-variants-proof} for any pair of such regions there are only two possibilities either the elements of the pair do not intersect or the smaller constituent is inside the bigger. 

Consider at first the case when the regions of the constituents $(i_1, l_1, i_2, l_2)$ and $(i'_1, l'_1, i'_2, l'_2)$ do not intersect. Since the segment $[i'_1; l_2]$ is a subset of the segment $[i_1; l_2]$ it is possible only when $l_1 \leq i'_1 \leq l'_2 \leq i_2$, which is one of the alternatives provided by the present lemma. It is illustrated on the picture below.

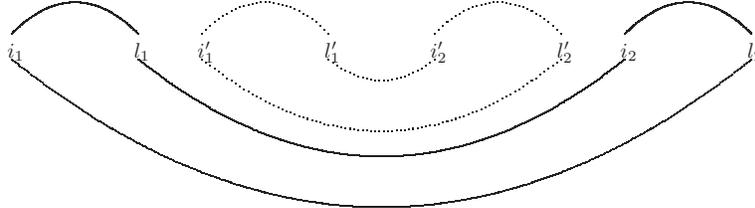
\scalebox{0.8}{
\begin{picture}(400,80)

\put(8,49){$i_1$}\put(68,49){$l_1$}\put(298,49){$i_2$}\put(358,49){$l_2$}
\qbezier(10,60)(40,90)(70,60)\qbezier(300,60)(330,90)(360,60)
\qbezier(10,48)(185,-92)(360,48)\qbezier(70,48)(185,-44)(300,48)
\put(98,49){$i'_1$}\put(158,49){$l'_1$}\put(208,49){$i'_2$}\put(268,49){$l'_2$}
\qbezier[30](100,60)(130,90)(160,60)\qbezier[30](210,60)(240,90)(270,60)
\qbezier[85](100,48)(185,-20)(270,48)\qbezier[25](160,48)(185,28)(210,48)
\end{picture}

}

\vspace*{12pt}
In the other case the region corresponding to the constituent $(i_1, l_1, i_2, l_2)$ contains all the other regions.
We consider different variants of embedding of the constituents. If constituents $(j_1, k_1, j_2, k_2)$ and $(i'_1, l'_1, i'_2, l'_2)$ do not intersect,then either $l'_2 \leq j_1$, $k_2 \leq i'_1$, $l'_1 \leq j_1 \leq k_2 \leq i'_2$ or $k_1 \leq i'_1 \leq l'_2 \leq j_2$. In the first case $i_1 \leq i'_1 \leq l'_2 \leq j_1$, symmetrically in the second $k_2 \leq i'_1 \leq l'_2 \leq l_2$, and in the third case $i_1 \leq i'_1 \leq l'_1 \leq j_1 \leq k_2 \leq i'_2 \leq l'_2 \leq l_2$ which all satisfy the requirements of the present lemma. The third case is illustrated on the picture below:

\vspace*{4pt}
\scalebox{0.8}
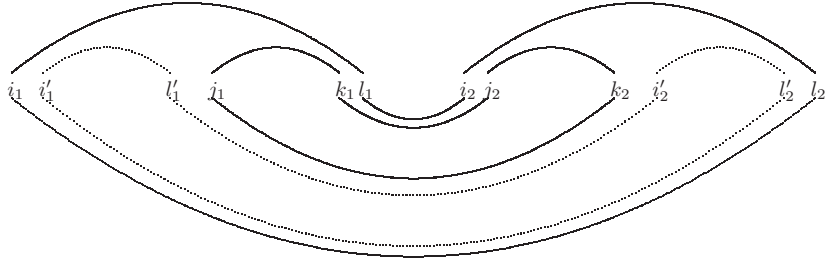
{
\begin{picture}(400,180)

\put(8,139){$i_1$}\put(174,139){$l_1$}\put(222,139){$i_2$}\put(388,139){$l_2$}
\qbezier(10,150)(93,216)(176,150)\qbezier(224,150)(307,216)(390,150)
\qbezier(10,138)(200,-12)(390,138)\qbezier(176,138)(200,118)(224,138)
\put(103,139){$j_1$}\put(163,139){$k_1$}\put(233,139){$j_2$}\put(293,139){$k_2$}
\qbezier(105,150)(135,174)(165,150)\qbezier(235,150)(265,174)(295,150)
\qbezier(105,138)(200,62)(295,138)\qbezier(165,138)(200,110)(235,138)
\put(23,139){$i'_1$}\put(83,139){$l'_1$}\put(313,139){$i'_2$}\put(373,139){$l'_2$}
\qbezier[40](25,150)(55,174)(85,150)\qbezier[40](315,150)(345,174)(375,150)
\qbezier[220](25,138)(200,-2)(375,138)\qbezier[140](85,138)(200,46)(315,138)
\end{picture}
}

\vspace*{-48pt}
Consider the last subcase $k_1 \leq i'_1 \leq l'_2 \leq j_2$, then applying Lemma \ref{const-variants-proof} to the constituents $(i_1, l_1, i_2, l_2)$ and $(i'_1, l'_1, i'_2, l'_2)$ we obtain that either $l'_2 \leq l_1$, $i_2 \leq i'_1$ or $l'_1 \leq l_1 \leq i_2 \leq i'_2$. Taking into account all the inequalities, we obtain that there are three possibilities: $k_1 \leq i'_1 \leq l'_2 \leq l_1$, $i_2 \leq i'_1 \leq l'_2 \leq j_2$ or $k_1 \leq i'_1 \leq l'_1 \leq l_1 \leq i_2 \leq i'_2 \leq l'_2 \leq j_2$. All these variants are allowed in the lemma statement. The latter variant is illustrated on the picture below:

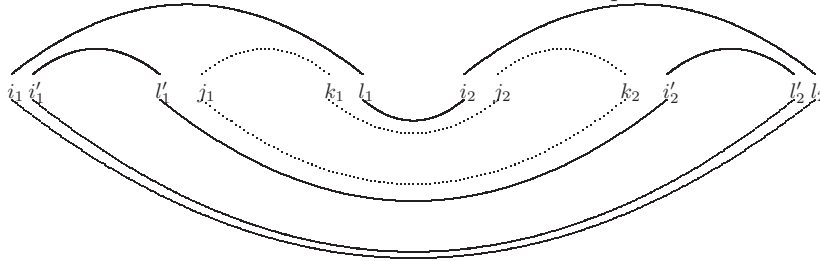
\scalebox{0.8}{
\begin{picture}(400,180)

\put(8,139){$i_1$}\put(174,139){$l_1$}\put(222,139){$i_2$}\put(388,139){$l_2$}
\qbezier(10,150)(93,216)(176,150)\qbezier(224,150)(307,216)(390,150)
\qbezier(10,138)(200,-12)(390,138)\qbezier(176,138)(200,118)(224,138)
\put(18,139){$i'_1$}\put(78,139){$l'_1$}\put(318,139){$i'_2$}\put(378,139){$l'_2$}
\qbezier(20,150)(50,174)(80,150)\qbezier(320,150)(350,174)(380,150)
\qbezier(20,138)(200,-6)(380,138)\qbezier(80,138)(200,42)(320,138)
\put(98,139){$j_1$}\put(158,139){$k_1$}\put(238,139){$j_2$}\put(298,139){$k_2$}
\qbezier[30](100,150)(130,174)(160,150)\qbezier[30](240,150)(270,174)(300,150)
\qbezier[100](100,138)(200,58)(300,138)\qbezier[40](160,138)(200,106)(240,138)

\end{picture}

}

\vspace*{-48pt}
Now consider the case when the region of the constituent $(i'_1, l'_1, i'_2, l'_2)$ is inside the region of $(j_1, k_1, j_2, k_2)$. It means that one of the following possibilities hold: $j_1 \leq i'_1 \leq k'_2 \leq k_1$, $j_2 \leq i'_1 \leq l'_2 \leq k_2$ or $i'_1 \leq j_1 \leq k_1 \leq l'_1 \leq i'_2 \leq j_1 \leq k_2 \leq l'_2$. All these variants satisfy the requirements of the Lemma.

So it remains to inspect the case when the region of constituent $(i'_1, l'_1, i'_2, l'_2)$ includes the region of $(j_1, k_1, j_2, k_2)$. This situation is illustrated on the picture below:

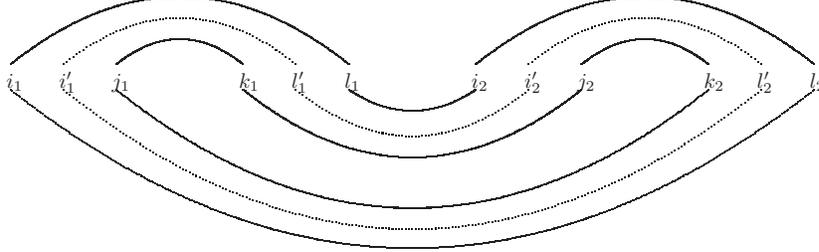
\scalebox{0.8}{
\begin{picture}(400,180)

\put(8,139){$i_1$}\put(168,139){$l_1$}\put(228,139){$i_2$}\put(388,139){$l_2$}
\qbezier(10,150)(90,214)(170,150)\qbezier(230,150)(310,214)(390,150)
\qbezier(10,138)(200,-12)(390,138)\qbezier(170,138)(200,118)(230,138)
\put(33,139){$i'_1$}\put(143,139){$l'_1$}\put(253,139){$i'_2$}\put(363,139){$l'_2$}
\qbezier[70](35,150)(90,194)(145,150)\qbezier[70](255,150)(310,194)(365,150)
\qbezier[210](35,138)(200,6)(365,138)\qbezier[70](145,138)(200,94)(255,138)
\put(58,139){$j_1$}\put(118,139){$k_1$}\put(278,139){$j_2$}\put(338,139){$k_2$}
\qbezier(60,150)(90,174)(120,150)\qbezier(280,150)(310,174)(340,150)
\qbezier(60,138)(200,26)(340,138)\qbezier(120,138)(200,74)(280,138)

\end{picture}

}

\vspace*{-36pt}
Then $i_1 \leq i'_1 \leq j_1 \leq k_1 \leq l'_1 \leq l_1 \leq i_2 \leq i'_2 \leq j_2 \leq k_2 \leq l'_2 \leq l_2$ and we should consider the mutual positions of the regions of constituents $(j_1, k_1, j_2, k_2)$ and $(j'_1, k'_1, j'_2, k'_2)$. This leads us to the following variants:
$$
\begin{array}{l}
i'_1 \leq j'_1 \leq k'_1 \leq j_1 \leq k_1 \leq l'_1 \leq i'_2 \leq j_2 \leq k_2 \leq j'_2 \leq k'_2 \leq l'_2,\\
i'_1 \leq j_1  \leq j'_1 \leq k'_1 \leq k_1 \leq l'_1 \leq i'_2 \leq j_2  \leq j'_2 \leq k'_2 \leq k_2 \leq l'_2,\\
i'_1 \leq j_1 \leq k_1  \leq j'_1 \leq k'_1 \leq l'_1 \leq i'_2  \leq j'_2 \leq k'_2 \leq j_2 \leq k_2 \leq l'_2,\\
i'_1 \leq j'_1 \leq j_1 \leq k_1 \leq k'_1 \leq l'_1 \leq i'_2 \leq j'_2 \leq j_2 \leq k_2 \leq \leq k'_2 \leq l'_2,
\end{array}
$$
But all such variants are allowed by the conclusion of the lemma. All the cases have been verified and the lemma is proved.
\end{proof}

\end{document}